\documentclass[11pt]{article}
\usepackage[T1]{fontenc}
\usepackage{graphicx} 
\usepackage{xcolor}
\usepackage{amsmath}
\usepackage{amsthm}
\usepackage{amssymb}
\usepackage[margin=1in]{geometry}
\usepackage[colorlinks=true, allcolors=blue]{hyperref}
\usepackage{natbib}

\title{A Tight Bound on Online Vertex Cover under Edge Arrivals}
\author{
Zhihao Gavin Tang\thanks{
Email: \url{tang.zhihao@mail.shufe.edu.cn}
}\\
Shanghai University of Finance and Economics
\and
Yuhao Zhang\thanks{
Email: \url{zhang_yuhao@sjtu.edu.cn}
}\\
Shanghai Jiao Tong University
}
\date{}

\newtheorem{theorem}{Theorem}
\newtheorem{lemma}{Lemma}[section]

\newtheorem{definition}[lemma]{Definition}
\newcommand{\ALG}{\textsc{Alg}}
\newcommand{\OPT}{\textsc{Opt}}

\renewcommand{\L}{\mathcal{L}}
\newcommand{\R}{\mathcal{R}}
\newcommand{\E}{\mathcal{E}}

\begin{document}

\maketitle

\begin{abstract}
We prove a tight impossibility result for online vertex cover under edge
arrivals.  No randomized integral or fractional algorithm achieves a
competitive ratio strictly below $2$ against an oblivious adversary, even on
bipartite graphs.  Since the standard algorithm that takes both endpoints of
every uncovered edge is $2$-competitive, this settles the optimal ratio.  Our
proof is a direct reduction from the recent breakthrough blueprint framework
of Assadi, Jiang, and Xiang.
\end{abstract}

\section{Introduction}
\label{sec:introduction}

In the edge-arrival online vertex cover problem, the vertex set is known in
advance and edges arrive one at a time.  The algorithm maintains a set of
selected vertices that must cover every edge seen so far.  This set can only
grow: selected vertices can never be removed.  
Adding both endpoints of each
uncovered edge is $2$-competitive because the triggering edges form a matching~\citep{DemangePaschos05}.
Can fractional decisions or randomization improve this factor?

Primal--dual algorithms yield the same competitive ratios for online matching and online vertex cover. Throughout this paper, we state competitive ratios using the convention that they are at least $1$.\footnote{The matching literature typically uses the reciprocal convention. Under that convention, the ratios discussed here are $1-1/e$, $\Gamma^*\approx 0.526$, and $1/2$.}
Under one-sided vertex arrivals, both problems admit the tight competitive ratio $e/(e-1)$ \citep{KVV90}; for vertex cover, the corresponding lower bound follows from ski rental~\citep{WangWong15}. Under general vertex arrivals, the best algorithms for both problems achieve a ratio of $1/\Gamma^*\approx 1.901$ \citep{WangWong15,TangZhang24}, and \citet{Tang26} recently proved this ratio tight for matching, even on bipartite graphs. Under edge arrivals, the standard algorithms achieve a ratio of $2$ for both problems, and \citet{GamlathEtAl19} proved this ratio tight for matching, again even on bipartite graphs.

These common upper bounds arise because the corresponding algorithms construct matching and vertex-cover solutions simultaneously. However, lower bounds for matching do not automatically transfer to vertex cover. Even in the one-sided vertex-arrival model, where the two problems have the same tight ratio, the vertex-cover lower bound comes from a completely different source, namely, a reduction from ski rental. Before our work, only in the one-sided vertex-arrival model was the algorithmic ratio matched by a vertex-cover lower bound. Under general vertex arrivals, the best known lower bound for vertex cover is only about $1.753$ \citep{WangWong15}. Under edge arrivals, it remained open whether a fractional or randomized algorithm could achieve a competitive ratio strictly smaller than $2$.

\paragraph{Our result.}
We close this gap.  For every deterministic fractional online vertex cover
algorithm and every $\varepsilon>0$, we construct a finite bipartite instance
with competitive ratio at least $2-\varepsilon$.  Taking expectations
coordinatewise gives the same bound for randomized integral and fractional
algorithms against an oblivious adversary.  Thus the optimal ratio is exactly
$2$.

\paragraph{Reduction and proof idea.}
Our key technical contribution is a direct reduction from proper blueprints,
introduced by Assadi, Jiang, and Xiang \citep{AssadiJiangXiang26a}.  From a
blueprint of value $\nu$, cyclic shifts hide an independent set $I$ of size
$2\nu n$.  Matching batches force at least $\nu n$ weight onto $I$ before it is
identified.  We then attach a private leaf to every vertex outside $I$.  The
offline optimum takes these $2(1-\nu)n$ vertices, while the algorithm also
retains its weight on $I$, giving ratio $(2-\nu)/(2(1-\nu))$.  Existing
blueprints have $\nu\to 2/3$ \citep{AssadiJiangXiang26b}, so this ratio tends to
$2$.

\paragraph{Relation to semi-streaming.}
Whether the factor $2$ can be improved for streaming vertex cover was raised
explicitly in prior work \citep{AssadiEtAl19EDCS,DerakhshanEtAl25}; the latter
also shows that constant-party one-way communication cannot establish a
factor-$2$ streaming lower bound.  The two blueprint papers derive
semi-streaming vertex cover lower bounds:
approximately $1.790$ in general graphs and the tight factor $2$ in bipartite
graphs \citep{AssadiJiangXiang26a,AssadiJiangXiang26b}, but they do not imply our
result.  Semi-streaming algorithms can postpone their output but have limited
memory; online algorithms have unlimited memory but make irrevocable choices.
Technically, the streaming reductions use ERS graph expansions, communication
complexity, and information theory; ours applies a blueprint directly.
Conceptually, streaming pays for forgetting absent edges, while online
computation pays for being unable to retract weight placed before the hidden
independent set is revealed.

\paragraph{Positive results in restricted settings.}
Positive results have been attained with extra assumptions on the input. For graphs of maximum degree $d$, \citet{BuchbinderSegevTkach17} gave a primal--dual algorithm with competitive ratio $2-2^{1-d}$ for both online fractional matching and online fractional vertex cover. On forests, they also obtained a randomized $9/5$-competitive matching algorithm. The dual cover used in that analysis, however, is constructed only after the full input is known and therefore does not yield an online vertex-cover algorithm. More recently, \citet{BaligacsEtAl26} gave an $11/6$-competitive algorithm for online fractional vertex cover on forests. Under a different restriction, \citet{LeeSingla20} obtained ratio $2-2^{1-s}$ for fractional bipartite vertex cover with $s$ arrival batches, where each batch is revealed in full before the corresponding irrevocable update.

\section{Preliminaries}

We first formally define the edge-arrival online vertex cover problem. The vertex set $V$ is known in advance, while the edges between these vertices arrive online. After the $t$-th edge arrives, the algorithm maintains a set $S_t\subseteq V$ that covers every edge revealed so far. The sets must be increasing: $S_{t-1}\subseteq S_t$. Thus, selected vertices can never be removed. In particular, if a new edge $(u,v)$ has neither endpoint in $S_{t-1}$, the algorithm must add at least one of $u$ and $v$. The objective is to minimize the size of the final set.

\paragraph{Fractional relaxation.}
In the fractional relaxation, after the $t$-th arrival the algorithm maintains a vector $x^{(t)}\in[0,1]^V$. These vectors are coordinatewise nondecreasing: $x^{(t-1)}_v\leq x^{(t)}_v$ for every $v\in V$. They must also satisfy $x^{(t)}_u+x^{(t)}_v\geq 1$ for every edge $(u,v)$ revealed by time $t$. Writing $x$ for the final vector, the objective is to minimize $\sum_{v\in V}x_v$.

The integral problem is a special case of the fractional problem obtained by restricting $x\in \{0,1\}^V$. Moreover, against an oblivious adversary, any randomized integral or fractional algorithm can be simulated by a deterministic fractional algorithm, without regard to computational efficiency, whose objective value equals the expected objective value of the randomized algorithm.

We evaluate online algorithms using competitive analysis. A fractional online vertex cover algorithm $\ALG$ is said to be $\Gamma$-competitive if, for every input sequence $\sigma$,
$$
\frac{\ALG(\sigma)}{\OPT(\sigma)} \leq \Gamma,
$$
where $\ALG(\sigma)$ denotes the objective value achieved by $\ALG$ on $\sigma$, and $\OPT(\sigma)$ denotes the objective value of an optimal offline solution for the same input.

\section{A Tight Lower Bound}
\label{sec:lower-bound}

In this section, we prove that no fractional online algorithm can achieve a competitive ratio strictly smaller than $2$ for the edge-arrival online vertex cover problem. By the equivalence discussed in the previous section, the same lower bound applies to randomized integral and fractional algorithms against an oblivious adversary.

\begin{theorem}
\label{thm:main}
For every deterministic fractional online vertex cover algorithm $\ALG$ and every $\varepsilon>0$, there exists a finite input sequence $\sigma$ such that
$$
\frac{\ALG(\sigma)}{\OPT(\sigma)}
\geq 2-\varepsilon.
$$
Moreover, the underlying graph of $\sigma$ is bipartite.
\end{theorem}

\subsection{Proof Overview}

We prepare two vertex sets $L_0$ and $R_0$ with $|L_0|=|R_0|=n$, and reveal a bipartite edge set $E_0$ between them. The main part of the construction produces a hidden set
$$
I=L_I\cup R_I,
$$
where $L_I\subseteq L_0$, $R_I\subseteq R_0$, and $|L_I|=|R_I|=\phi$, with the following two properties:
\begin{enumerate}
    \item $I$ is an independent set in $(L_0\cup R_0,E_0)$.
    \item After processing $E_0$, the fractional solution maintained by $\ALG$ satisfies
    $$
    \sum_{v\in I}x_v\geq \phi=\frac{|I|}{2}.
    $$
\end{enumerate}

We say that the algorithm \emph{fully covers} a vertex set $U$ if $\sum_{v\in U}x_v\geq |U|/2$. Thus, the second property states that $I$ is fully covered.

The term ``hidden'' means that the identities of the vertices in $I$ are concealed from the algorithm while the edges in $E_0$ arrive. Before the algorithm can determine the hidden set, it may have already assigned substantial irrevocable fractional weight to its vertices. This unnecessary weight is the main source of the lower bound.

For every vertex $v\in L_0\cup R_0$, we reserve a distinct auxiliary vertex $v'$. All auxiliary vertices are included in the vertex set known to the algorithm from the beginning, but remain isolated while $E_0$ is revealed. Once the hidden set has been determined, we reveal a pendant edge $(v,v')$ for every $v\notin I$.

The following lemma formalizes this completion step.

\begin{lemma}
\label{lem:completion}
Suppose that $E_0$ is a bipartite edge set on $L_0\cup R_0$, where $|L_0|=|R_0|=n$. Suppose further that, after processing $E_0$, there exists an independent set $I=L_I\cup R_I$ satisfying
$$
|L_I|=|R_I|=\phi
\qquad\text{and}\qquad
\sum_{v\in I}x_v\geq \phi.
$$
Then the input can be extended, while preserving bipartiteness, so that
$$
\frac{\ALG}{\OPT}
\geq
\frac{2n-\phi}{2(n-\phi)}.
$$
\end{lemma}

\begin{proof}
For every $v\in(L_0\cup R_0)\setminus I$, reveal the pendant edge $(v,v')$. Let $E_1$ denote the set of these edges, and place each auxiliary vertex $v'$ on the side opposite to $v$.

Since $I$ is independent in $(L_0\cup R_0,E_0)$, the set $(L_0\cup R_0)\setminus I$ covers every edge in $E_0\cup E_1$. Therefore,
$$
\OPT
\leq
|(L_0\cup R_0)\setminus I|
=
2(n-\phi).
$$

The edges in $E_1$ are pairwise vertex-disjoint, so their covering constraints force the algorithm to assign total weight at least $2(n-\phi)$ to their endpoints. Moreover, the fractional weight already assigned to $I$ cannot decrease. Since the endpoints of $E_1$ are disjoint from $I$, we obtain
$$
\ALG
\geq
2(n-\phi)+\phi
=
2n-\phi.
$$
Consequently,
$$
\frac{\ALG}{\OPT}
\geq
\frac{2n-\phi}{2(n-\phi)}.
$$
\end{proof}

Thus, it remains to construct an independent set with $\phi$ close to $2n/3$ that is fully covered.

\subsection{The Blueprint Construction}

To construct such a hidden independent set, we use the blueprint technique introduced in \cite{AssadiJiangXiang26a}.

\begin{definition}[Proper blueprint]
Let $P\ge 1$ and $C\ge 2$ be integers. A \emph{blueprint} $\mathcal{B}=(\L,\R,\E)$ consists of two vertex sets
$$
\L=\{\L(a):a\in[C]^P\}
\qquad\text{and}\qquad
\R=\{\R(a):a\in[C]^P\},
$$
together with a bipartite edge set partitioned as
$$
\E
=
\E^{(1)}
\mathbin{\dot{\cup}}
\cdots
\mathbin{\dot{\cup}}
\E^{(P)}.
$$
For $a\in[C]^P$, let $a_{<p}$ denote the prefix consisting of its first $p-1$ coordinates, and let $\circ$ denote concatenation. \footnote{For two vectors $u=(u_1,\ldots,u_k)$ and $v=(v_1,\ldots,v_\ell)$, their concatenation is $u\circ v=(u_1,\ldots,u_k,v_1,\ldots,v_\ell)$.}

The blueprint is \emph{proper} if the following conditions hold:
\begin{enumerate}
    \item $\E$ is a matching.
    
    \item For every edge $(\L(a),\R(b))\in\E^{(p)}$ and every $c\in[C]^{P-p+1}$, at least one of $\L(a_{<p}\circ c)$ and $\R(b_{<p}\circ c)$ is isolated in the blueprint.
\end{enumerate}
Here, a blueprint vertex is \emph{isolated} if it is not incident to any edge in $\E$. The \emph{value} of the blueprint is
$$
\nu(\mathcal{B})=\frac{|\E|}{C^P}.
$$
\end{definition}

We use the following result.

\begin{theorem}[Blueprint theorem~{\cite[Theorem~3]{AssadiJiangXiang26b}}]
\label{thm:blueprint}
There exists a sequence of finite proper blueprints whose values converge to $2/3$.
\end{theorem}

Fix a finite proper blueprint $\mathcal{B}=(\L,\R,\E)$ with parameters $P$ and $C$. We prepare two vertex sets
$$
L_0=\{L_0(a):a\in[C]^P\}
\qquad\text{and}\qquad
R_0=\{R_0(a):a\in[C]^P\}.
$$
Thus, $|L_0|=|R_0|=C^P$. The symbols $L_0(a)$ and $R_0(a)$ refer to vertices in the online instance, whereas $\L(a)$ and $\R(a)$ refer to vertices in the blueprint.

We would like to use the endpoints of the blueprint edges as the hidden independent set. However, mapping $\L(a)$ directly to $L_0(a)$ and $\R(b)$ directly to $R_0(b)$ would reveal the hidden set from the beginning. To conceal this mapping, we use a cyclic shift.

Throughout the construction, we identify $[C]$ with the cyclic group $\mathbb{Z}_C$, and use $\oplus$ and $\ominus$ to denote coordinate-wise addition and subtraction modulo $C$. For every shift $z\in[C]^P$, we map
$$
\L(a)\longmapsto L_0(a\oplus z)
\qquad\text{and}\qquad
\R(b)\longmapsto R_0(b\oplus z).
$$
The hidden set associated with $z$ is
$$
I(z)
=
\bigcup_{(\L(a),\R(b))\in\E}
\{L_0(a\oplus z),R_0(b\oplus z)\}.
$$
Since $\E$ is a matching, $I(z)$ contains exactly $|\E|$ vertices from each side. In particular,
$$
|I(z)|=2|\E|.
$$

We next define the edge sequence associated with a fixed shift $z$. For every phase $t\in[P]$ and every blueprint edge
$$
e=(\L(a),\R(b))\in\E^{(t)},
$$
the adversary reveals the batch
$$
B_e^{(t)}(z)
=
\left\{
\left(
L_0\bigl((a\oplus z)_{<t}\circ c\bigr),
R_0\bigl((b\oplus z)_{<t}\circ c\bigr)
\right)
:
c\in[C]^{P-t+1}
\right\}.
$$
Each batch forms a perfect matching between two corresponding subcubes. Notice that $B_e^{(t)}(z)$ depends only on the prefix $z_{<t}$.

The phases are revealed in the order $1,2,\ldots,P$. Within each phase, the batches and their edges are revealed in an arbitrary fixed order. If an edge has already appeared in an earlier batch, it need not be revealed again. Let $E_0(z)$ denote the set of all distinct edges revealed during the construction.

For the analysis, let $Z$ be chosen uniformly at random from $[C]^P$. Uppercase $Z$ always denotes this random shift, whereas lowercase $z$ denotes a fixed realization. Let $x^{(t)}$ denote the solution maintained by $\ALG$ after phase $t$, and let $x$ denote its solution after all $P$ phases. When $Z$ is random, these vectors are also random.

\paragraph{The Toy Example.}
We illustrate the construction using a simple example. Let $C=2$ and $P=1$, and identify $[2]$ with $\{0,1\}$. Consider the blueprint
$$
\E=\E^{(1)}=\{(\L(0),\R(1))\}.
$$
This blueprint is proper. For $c=0$, the vertex $\R(0)$ is isolated, while for $c=1$, the vertex $\L(1)$ is isolated. Its value is
$$
\nu(\mathcal{B})
=
\frac{|\E|}{C^P}
=
\frac{1}{2}.
$$

Since $P=1$, the batch associated with the only blueprint edge contains all possible labels in $[2]$. Therefore, regardless of the shift $z$, the adversary reveals
$$
B_e^{(1)}(z)
=
\{(L_0(0),R_0(0)),(L_0(1),R_0(1))\}.
$$

The shift determines which endpoints form the hidden set. If $z=0$, then
$$
I(0)=\{L_0(0),R_0(1)\},
$$
whereas if $z=1$, then
$$
I(1)=\{L_0(1),R_0(0)\}.
$$
Thus, the two possible shifts correspond exactly to the two candidate hidden independent sets.

Moreover,
$$
\begin{aligned}
\mathbb{E}_Z
\left[
\sum_{v\in I(Z)}x_v
\right]
&=
\frac{1}{2}
\left(
x_{L_0(0)}
+x_{R_0(1)}
+x_{L_0(1)}
+x_{R_0(0)}
\right) \\
&=
\frac{1}{2}
\left(
x_{L_0(0)}
+x_{R_0(0)}
+x_{L_0(1)}
+x_{R_0(1)}
\right) \\
&\geq 1.
\end{aligned}
$$
The inequality follows from the covering constraints of the two revealed edges. Hence, one of the two shifts produces a hidden independent set that is fully covered.

This example captures the two roles of a proper blueprint. The released batch forces sufficient fractional weight onto the hidden endpoints, while the properness condition guarantees that the hidden endpoints form an independent set.

\subsection{Analysis}

We first show that the batch corresponding to each blueprint edge forces one unit of expected fractional weight onto its two hidden endpoints.

\begin{lemma}
\label{lem:expected-mass}
For every phase $t$ and every edge $(\L(a),\R(b))\in\E^{(t)}$,
$$
\mathbb{E}_Z
\left[
x^{(t)}_{L_0(a\oplus Z)}
+
x^{(t)}_{R_0(b\oplus Z)}
\right]
\geq 1.
$$
\end{lemma}

\begin{proof}
Condition on the prefix $Z_{<t}$. All edges revealed through phase $t$ depend only on $Z_{<t}$. Hence, under this conditioning, the input observed by $\ALG$ and the resulting vector $x^{(t)}$ are fixed.

The suffix $Z_{\geq t}$ remains uniformly distributed over $[C]^{P-t+1}$. Therefore, $a_{\geq t}\oplus Z_{\geq t}$ and $b_{\geq t}\oplus Z_{\geq t}$ are each uniformly distributed over $[C]^{P-t+1}$. Reindexing the two marginal sums separately and using linearity of expectation, we obtain
\begin{align*}
    &\mathbb{E}_{Z_{\geq t}}
    \left[
        x^{(t)}_{L_0(a\oplus Z)}
        +
        x^{(t)}_{R_0(b\oplus Z)}
        \,\middle|\,
        Z_{<t}
    \right] \\
    ={}&
    \frac{1}{C^{P-t+1}}
    \sum_{c\in[C]^{P-t+1}}
    \left(
        x^{(t)}_{L_0((a\oplus Z)_{<t}\circ c)}
        +
        x^{(t)}_{R_0((b\oplus Z)_{<t}\circ c)}
    \right).
\end{align*}
For every $c\in[C]^{P-t+1}$, the pair
$$
\left(
    L_0((a\oplus Z)_{<t}\circ c),
    R_0((b\oplus Z)_{<t}\circ c)
\right)
$$
belongs to the batch $B_e^{(t)}(Z)$. By the end of phase $t$, this edge has been revealed. Therefore, feasibility of the fractional vertex cover implies
$$
x^{(t)}_{L_0((a\oplus Z)_{<t}\circ c)}
+
x^{(t)}_{R_0((b\oplus Z)_{<t}\circ c)}
\geq 1.
$$
It follows that
$$
\mathbb{E}_{Z_{\geq t}}
\left[
    x^{(t)}_{L_0(a\oplus Z)}
    +
    x^{(t)}_{R_0(b\oplus Z)}
    \,\middle|\,
    Z_{<t}
\right]
\geq 1.
$$
Taking the expectation over $Z_{<t}$ proves the lemma.
\end{proof}

The properness condition ensures that the hidden set remains independent.

\begin{lemma}
\label{lem:hidden-independent}
For every shift $z\in[C]^P$, the set $I(z)$ is independent in $(L_0\cup R_0,E_0(z))$.
\end{lemma}

\begin{proof}
Consider an edge in the batch associated with $(\L(a),\R(b))\in\E^{(t)}$. It has the form
$$
\left(
L_0\bigl((a\oplus z)_{<t}\circ c\bigr),
R_0\bigl((b\oplus z)_{<t}\circ c\bigr)
\right)
$$
for some $c\in[C]^{P-t+1}$.

Under the inverse shift, its endpoints correspond to
$$
\L(a_{<t}\circ c')
\qquad\text{and}\qquad
\R(b_{<t}\circ c'),
$$
where $c'=c\ominus z_{\geq t}$. By properness, at least one of these two blueprint vertices is isolated. Its image under the shift is therefore not contained in $I(z)$.

Hence, every revealed edge has at least one endpoint outside $I(z)$, and $I(z)$ is independent.
\end{proof}

We can now combine the expected contributions of all blueprint edges.

\begin{lemma}
\label{lem:hidden-covered}
There exists a shift $z\in[C]^P$ such that, after all $P$ phases,
$$
\sum_{v\in I(z)}x_v
\geq
\frac{|I(z)|}{2}.
$$
\end{lemma}

\begin{proof}
Since $\E$ is a matching, different blueprint edges correspond to distinct vertices in $I(Z)$. Moreover, fractional variables can only increase. Lemma~\ref{lem:expected-mass} therefore gives
$$
\begin{aligned}
\mathbb{E}_Z
\left[
\sum_{v\in I(Z)}x_v
\right]
&\geq
\sum_{t=1}^P
\sum_{(\L(a),\R(b))\in\E^{(t)}}
\mathbb{E}_Z
\left[
x^{(t)}_{L_0(a\oplus Z)}
+
x^{(t)}_{R_0(b\oplus Z)}
\right] \\
&\geq
\sum_{t=1}^P|\E^{(t)}|
=
|\E|
=
\frac{|I(Z)|}{2}.
\end{aligned}
$$
Since $|I(Z)|=2|\E|$ is independent of $Z$, some fixed shift $z$ satisfies the desired inequality.
\end{proof}

\subsection{Proof of the Main Theorem}

We now complete the proof of the main theorem.

\begin{proof}[Proof of Theorem~\ref{thm:main}]
By Theorem~\ref{thm:blueprint}, there exists a finite proper blueprint whose value $\nu$ is arbitrarily close to $2/3$. Choose such a blueprint so that
$$
\frac{2-\nu}{2(1-\nu)}
\geq
2-\varepsilon.
$$
Let $n=C^P$ and $\phi=|\E|=\nu n$.

By Lemma~\ref{lem:hidden-covered}, there exists a shift $z\in[C]^P$ for which $I(z)$ is fully covered. By Lemma~\ref{lem:hidden-independent}, the same set is independent. Lemma~\ref{lem:completion} therefore gives
$$
\frac{\ALG}{\OPT}
\geq
\frac{2n-\phi}{2(n-\phi)}
=
\frac{2-\nu}{2(1-\nu)}
\geq
2-\varepsilon.
$$
The chosen blueprint is finite, so the resulting input sequence is finite. The construction also preserves bipartiteness. This completes the proof.
\end{proof}

\section*{Declaration on the Use of AI Tools}

The authors began investigating this problem in February 2026 and made partial
progress beyond the previously known $1.753$ lower bound for the general
vertex-arrival model.  The authors then suggested that the recent blueprint
framework of Assadi, Jiang, and Xiang might apply.  Prompted by this suggestion,
OpenAI's GPT-5.6 Sol formulated the reduction yielding the tight factor-$2$
lower bound proved in this paper and assisted with drafting the manuscript.
The authors independently verified the reduction, proof, citations, and
exposition and can directly defend every aspect of the submission.  The authors
assume responsibility for all content.

\bibliographystyle{plainnat}
\bibliography{references}

\end{document}